\documentclass[11pt]{article}

\usepackage{fullpage}
\usepackage{amssymb,amsmath}
\usepackage{amsthm}
\usepackage{thmtools}
\usepackage{mathdots}

\usepackage{cases}
\usepackage{csquotes}
\usepackage{graphicx,psfrag}
\usepackage{microtype}
\usepackage{xspace}
\usepackage{marvosym}
\usepackage{tabularx}
\usepackage{color,soul}
\usepackage{cleveref}
\usepackage{paralist}

\usepackage[vlined,linesnumbered,ruled]{algorithm2e}
\usepackage{algorithmic}
\usepackage{float}

\usepackage{subcaption}

\usepackage{csquotes}

\usepackage{tikz}
\usetikzlibrary{positioning}
\usetikzlibrary{shapes}
\usetikzlibrary{shapes.symbols,patterns}
\usetikzlibrary{calc,through,backgrounds}
\usetikzlibrary{decorations.pathreplacing}
\usetikzlibrary{arrows}
\usetikzlibrary{matrix}
\usetikzlibrary{shadows}

\parskip = \smallskipamount

\newtheorem{theorem}{Theorem}
\newtheorem{lemma}[theorem]{Lemma}

\newtheorem{observation}[theorem]{Observation}
\newtheorem{corollary}[theorem]{Corollary}

\theoremstyle{definition}
\newtheorem{example}{Example}

\newcommand{\paren}[1]{\left( #1 \right)}

\newcommand{\set}[1]{\left\{ #1 \right\}}

\newcommand{\ceil}[1]{\left\lceil #1 \right\rceil}
\newcommand{\floor}[1]{\left\lfloor #1 \right\rfloor}

\newcommand{\half}{\frac{1}{2}}

\newcommand{\multipl}{\omega}

\newcommand{\BC}{\text{BC}}

\newcommand{\calB}{\mathcal{B}}
\newcommand{\calF}{\mathcal{F}}

\title{\textbf{Degree Realization by Bipartite Cactus Graphs}%
\thanks{
An extended astract appeared in the proceedings of 
14th International Conference on Algorithms and Complexity (CIAC), 2025.
This work was supported by US-Israel BSF grant 2022205.}
}

\author{%
Amotz Bar-Noy%
\thanks{Department of Computer Science, City University of New York, New York, USA;
{\tt amotz@sci.brooklyn.cuny.edu}}
\and
Toni B\"ohnlein\thanks{Huawei, Zurich, Switzerland;
{\tt toniboehnlein@web.de}}
\and
David Peleg\thanks{Department of Computer Science and Applied Mathematics, 
Weizmann Institute of Science, Rehovot, Israel;
{\tt david.peleg@weizmann.ac.il}}
\and
Yingli Ran\thanks{College of Mathematics and Computer Science, 
Zhejiang Normal University, China;
{\tt ranyingli@zjnu.edu.cn}}
\and
Dror Rawitz\thanks{Faculty of Engineering, Bar Ilan University, Ramat-Gan, Israel;
{\tt dror.rawitz@biu.ac.il}}
}

\begin{document}

\begin{titlepage}

\maketitle

\begin{abstract}
The \textsc{Degree Realization} problem with respect to a graph family
$\calF$ is defined as follows. The input is a sequence $d$ of $n$ positive
integers, and the goal is to decide whether there exists a graph $G \in \calF$
whose degrees correspond to $d$. The main challenges are to provide a
precise characterization of all the sequences that admit a realization in
$\calF$ and to design efficient algorithms that construct one of the possible
realizations, if one exists.

This paper studies the problem of realizing degree sequences by
bipartite cactus graphs (where the input is given as a single sequence,
without the bi-partition).
A characterization of the sequences that have a cactus realization is already
known~\cite{Rao81cactus}.
In this paper, we provide a systematic way to obtain such a characterization,
accompanied by a realization algorithm.
This allows us to derive a characterization for bipartite cactus graphs, and as a byproduct, 
also for several other interesting sub-families of cactus graphs, including
bridge-less cactus graphs and core cactus graphs, as well as for
the bipartite sub-families of these families.
\end{abstract}

\medskip


\renewcommand{\thepage}{}
\end{titlepage}
\pagenumbering{arabic}


\section{Introduction}
\label{sec:intro}


\subsection{Background and motivation}

We study graph realization problems in which for some specified graph family $\calF$,
a sequence of integers $d$ is given, and the requirement is to construct a graph
from $\calF$ whose degrees abide by $d$.
More formally, the \textsc{Degree Realization} problem with respect to a graph
family $\calF$ is defined as follows. The input is a sequence $d = (d_1,\ldots,d_n)$
of positive integers, and the goal is to decide whether there exists a graph $G \in \calF$
whose degrees correspond to $d$, i.e.,
where the vertex set is $V = \set{1,\ldots,n}$ and $\deg_G(i) = d_i$, for every $i\in V$.
If such a graph exists, then $d$ is called \emph{$\calF$-graphic}.
Observe that while every graph $G \in \calF$ corresponds to a unique degree sequence $d$,
a degree sequence may be realized by more than one graph in $\calF$. 
For example, in the family of bipartite graphs, the sequence $(2,2,2,2,2,2,2,2)$
can be realized by a $8$-vertex cycle or by two $4$-vertex cycles.

There are two fundamental challenges that arise in this context.
The first is to provide an algorithm that \emph{decides} whether a given sequence
can be realized by a graph from $\calF$,
and furthermore to provide a characterization of all the realizable sequences.
The second is to design an efficient algorithm that \emph{constructs} one
of the realizations if one exists.

We consider sub-families of the family of \emph{cactus graphs} (\emph{cacti}). 
A cactus graph is a connected graph in which any edge may be a member of at most one cycle,
which means that different cycles do not share edges, but may share one vertex.
Cacti are an important and interesting graph family
with many applications, for instance in modeling electric circuits~\cite{NishiChua86,WagnerWolf21},
communication networks~\cite{Arcak11}, and genome comparisons~\cite{PDEJMSH11}.
We provide a characterization for
the \textsc{Degree Realization} problem with respect to bipartite cactus graphs
(bi-cacti), which implies a linear time algorithm for the decision problem.
Furthermore, we provide a linear time
realization algorithm for degree realization by bipartite cactus graphs.
We introduce a systematic way to obtain such a characterization,
which allows us to obtain the known characterization for cactus graphs~\cite{Rao81cactus}
and in addition also some new (previously unknown) characterizations
for several other interesting sub-families of cactus graphs,
including bridge-less cactus graphs and core cactus graphs, as well as
the bipartite sub-families of these families.

The characterization of families of sparse graphs, such as cactus graphs, may 
assist in finding a characterization for the family of planar graphs and the family 
of outer-planar graphs, both of which are open problems for about half a century.


\subsection{Related Work}

\paragraph*{\bf General graphs.}
The \textsc{Degree Realization} problem with respect to the family of all graphs
was studied extensively in the past.
Erd\H{o}s and Gallai~\cite{EG60} gave a necessary and sufficient condition
(which also implies an $O(n)$ decision algorithm) for a sequence to be realizable, or \emph{graphic}.
Havel~\cite{havel55} and Hakimi~\cite{Hakimi62} (independently) gave another
characterization for graphic sequences, which also implies an efficient $O(m)$ time
algorithm for constructing a realizing graph for a given graphic sequence, where $m$
is the number of edges in the graph. A variant of this realization algorithm is given
in~\cite{WK73}.


\paragraph*{\bf Bipartite graphs.}
The history of the \textsc{Degree Realization} problem with respect to the family of
bipartite graphs is as long as the one for general graphs. In this problem, a sequence
is given as input and the goal is to find a realizing bipartite graph.
This problem has a variant in which the input consists of two sequences representing
the degree sequences of the two sides of a bipartite realization. This variant
was solved by Gale and Ryser~\cite{Gale57,Ryser57} even before Erd\H{o}s and Gallai's
characterization of graphic sequences. However, the general bipartite realization
problem remains open despite being mentioned as open over 40 years ago~\cite{Rao81}. 
Recent attempts solve special cases and emphasize \emph{approximate} 
realizations~\cite{BBPR22b,BBPR25-hl}.
The sequence $d$ is called \emph{forcibly $\calF$-graphic} if every realization of $d$
is in $\calF$. Characterizations of sequences that are forcibly bipartite-graphic 
or forcibly connected bipartite-graphic were given in~\cite{BBPR23}.


\paragraph*{\bf Sparse graphs.}
The most relevant category is of families which contain graphs with a linear number of edges.
%
The problem is straightforward with respect to trees~\cite{GJT07}, forests, and
unicyclic graphs~\cite{BoeschHarary1978}.
Characterizations of sequences that are forcibly forest
or forcibly tree were obtained in~{\cite{BBPR23}}.
Characterizations of forcibly unicyclic and bicyclic sequences were given in~{\cite{DuanTian25}}.
(A graph $G = (V,E)$ is unicyclic if it is connected and $|E| = n$; it is called bicyclic 
if it is connected and $|E| = n+1$.)
A characterization for Halin graphs was given in~\cite{Biyikoglu05}.
Rao~\cite{Rao81cactus} provided a characterization for cactus graphs and for forcibly cactus graphs.
He also gave a characterization for cactus graphs whose cycles are triangles
and for connected graphs whose blocks are cycles of $k$ vertices.
Beineke and Schmeichel~\cite{BeinekeSchmeichel79}
characterized cacti degree sequences with up to four cycles and also provided
a sufficient condition for cactus realization.

Rao~\cite{Rao81} mentioned \textsc{Degree Realization} with respect to planar graphs 
and related families as open.
A characterization is known for regular sequences~\cite{hawkins1966certain}
and for sequences with $d_1 - d_n = 1$, where $d$ is assumed to be in non-increasing order~\cite{SH77}.
Partial results are known if $d_1 - d_n = 2$~\cite{fanelli1981unresolved-nonmaximal,fanelli1980conjecture-maximal,SH77}.
A characterization of bi-regular sequences with respect to the family of bipartite planar graphs
is given in~\cite{AdamsN19}.
A sufficient condition for planarity was given in~\cite{BBPRR24-mfcs-planar}.
As for outerplanar graphs, only partial results are known.
Several necessary conditions were given in~\cite{syslo1979characterizations,Jao2010West}.
Choudum~\cite{Choudum81} gave a characterization for forcibly outerplanar sequences.
In~\cite{BBPRR25} it was shown that any sequence that satisfies a certain necessary condition for 
outerplanarity is either non-outerplanaric or has a 2-page book embedding.
A sufficient condition was given in~\cite{BBPRR24-mfcs-outerplanar}.
Sufficient conditions for the realization of $2$-trees were given in~\cite{LMNW06}.
(A graph $G$ is a 2-tree if $G$ is a triangle or $G$ has a degree-2 vertex
whose neighbors are adjacent and whose removal leaves a 2-tree.)
Bose et at.~\cite{BDKLMWW08} gave a characterization 
for $2$-trees with a linear time realization algorithm.


\subsection{Our Results and Techniques}

As opposed to the approach taken in~\cite{Rao81cactus}, 
the characterizations and realization algorithms of this paper were developed by starting
with simple graph families and gradually coping with families that are more involved.
Specifically, \Cref{sec:pseudotree} contains 
characterizations and realization algorithms for
\emph{unicyclic} graphs and \emph{bi-unicyclic} graphs. 
Coping with the simplest non-trivial cacti provides the basic techniques
needed for the more general cases, but this also serves as a light introduction to degree realization.
\emph{Bridge-less cacti} are studied in \Cref{sec:bridge-less}, which
contains characterizations and realization algorithms for \emph{bridge-less cactus} graphs and
\emph{bridge-less bi-cactus} graphs.
The next family we consider is that of \emph{core cactus} graphs (see definition in \Cref{sec:prelim}).
\Cref{sec:core} contains characterizations and realization algorithms for \emph{core cactus} graphs and
\emph{core bi-cactus} graphs.
In \Cref{sec:cactus} we provide a characterization for degree realization by
\emph{cactus} graphs and \emph{bi-cactus} graphs.
Finally, in \Cref{sec:forcibly} we give a characterization for forcibly bi-cactus 
and forcibly bipartite unicyclic graphs.

The crux in developing a necessary and sufficient condition for cactus and bi-cactus
realizability of a given sequence is to bound the number of possible edges
in the realizing graph as shown in \Cref{sec:bounds}.
This is obtained when the number of \emph{bridges}
in the graph (see definition in \Cref{sec:prelim}) is minimized.
The above condition depends on a \emph{bridge parameter},
which is defined as
\[
\beta \triangleq \max \set{\multipl_1 , \half(\multipl_1 + \multipl_{odd})  }
~,
\]
where $\multipl_1$ is the multiplicity of $1$ in $d$ and $\multipl_{odd}$
is the number of odd integers greater than $1$ in $d$.
We note that this parameter is implicit in~\cite{Rao81cactus}.
The decision about a given sequence depends only on the volume $\sum_i d_i$, $n$,
and the bridge parameter (see Theorem~\ref{thm:cactus}).
We believe that this parameter may be of independent interest.

The decision and realization algorithms of all the above mentioned families work in linear time.
Our algorithms are reminiscent of the minimum pivot version of the
Havel-Hakimi algorithm~\cite{havel55,Hakimi62} for realizing sequences by general graphs.
However, in our algorithms pivots are not connected to the vertices with the \emph{maximum} residual
degrees in the sequence $d$. Hence, our analysis is not based on swapping
arguments. In particular, degree-$1$ vertices should be connected
to odd degree vertices, rather than to even degree vertices, even if the latter degrees are larger.
When the sequence does not contain degree-$1$ vertices, pairs of degree-$2$
vertices are used to construct a triangle that lowers the degree of another vertex by $2$.
Again, smaller odd degree vertices are preferred over larger even degree vertices.
Throughout the paper, when dealing with bi-cactus graphs,
we adapt the techniques used in the cactus case to avoid constructing odd-length
cycles. This task turned out to be more involved, since in this case a realization
may require one extra bridge edge.

Given a graph family, realizability of sequences may depend on certain parameters. 
There are two extremes. One extreme is the elaborate test of Erd\H{o}s-Gallai that  
examines the relationship among all the degrees before determining if a sequence is graphic. 
The other extreme is for forests in which the length of the sequence $n$ 
and the sequence sum are the only two interesting parameters, 
i.e., a sequence $d$ is realizable by a forest if and only if $\sum_i d_i \leq 2n-2$. 
The results for cacti and bi-cacti are not that simple, but still 
depend only on four parameters: the multiplicities of 1's and of odd numbers, the sequence length,
and the sequence sum. Our structured proof demonstrates the roles of these 
two additional parameters, through the bridge parameter. 
A possible next step could be utilizing additional parameters, e.g., 
the multiplicity of 2's, to obtain characterizations of 
sequences that can be realized by other families of sparse graphs, such as planar graphs and 
outerplanar graphs. Both of which are long standing open problems.


\section{Preliminaries}
\label{sec:prelim}


\subsection{Definitions and Notation}

We consider simple graphs $G = (V,E)$, where $V = \set{1, \ldots, n}$.
The degree of a vertex $i \in V$, denoted by $\deg_G(i)$, is its number of  neighbors.
The degree sequence of graph $G$ is $\deg(G) = (\deg_G(1), \ldots, \deg_G(n))$.
Let $d = (d_1, \ldots,d_n)$ be a sequence of positive integers.
If there exists a graph $G$ such that $\deg(G) = d$, then it is said that $G$ \emph{realizes} $d$.
A sequence $d$ that has a realization $G$ is called $\emph{graphic}$.
We refer to $\sum_i d_i$ as the \emph{volume} of $d$. Define $m \triangleq \half \sum_i d_i$.
Notice that if $d$ is graphic, then $m$ is the number of edges in any realization of $d$.
A sequence $d$ is called a \emph{degree sequence} if $d_i \in \set{1,\ldots,n-1}$,
for every $i$, and the volume $\sum_i d_i$ is even.
Throughout the paper, we assume that $d_i \geq d_{i+1}$, for every $1\le i\le n-1$.
For brevity, we use $a^k$ as a shorthand for a subsequence of $k$ consecutive $a$'s (e.g., $2^3$ represents $2,2,2$).
Given a degree sequence $d$, let $\multipl_i$ be the number of times the integer $i$ appears in $d$. 
Finally, $\multipl_{odd}$ is the number of odd integers that are larger than 1 in $d$,
namely $\multipl_{odd} = \sum_{k \geq 1} \multipl_{2k+1}$.
Consider for example the degree sequence $(9,5^5,4^2,3^4,2,1^8)$. 
For this sequence, $\multipl_1 = 8, \multipl_2 = 1, \multipl_3 = 4, \multipl_4 = 2, 
\multipl_5 = 5, \multipl_7 = \multipl_8 = 0, \multipl_9 = 1$.
Also, $\multipl_{odd} = \multipl_3 + \multipl_5 + \multipl_7 +\multipl_9 = 10$.


\subsection{Graph Families}

A graph $G$ is \emph{connected} if there is a path between every pair of vertices in the graph.
A \emph{cut-vertex} of a connected graph is a vertex whose removal disconnects the graph.
A \emph{bridge} in a connected graph is an edge whose removal disconnects the graph.
A \emph{block} of $G$ is a maximal connected subgraph of $G$ that does not have cut-vertices.
That is, it is a maximal subgraph which is either an isolated vertex, a bridge edge,
or a 2-connected subgraph.

A graph $G$ is called a \emph{pseudo-tree} if it is connected and it contains at most one cycle.
It is called \emph{unicyclic} if it contains exactly one cycle.
A graph $G$ is called a \emph{pseudo-forest} if each of its connected components is a pseudo-tree.
A \emph{cactus} graph is a connected graph in which any edge may 
be a member of at most one cycle,
which means that different cycles do not share edges, but may share one vertex.
An alternative definition is that a graph $G$ is a (non-trivial) cactus 
if and only if every block
of $G$ is either a simple cycle or a bridge (see \Cref{fig:cactus}).
A cactus $G$ is called \emph{bridge-less} if it has no bridges.
In this case every edge belongs to exactly one cycle in $G$ (see \Cref{fig:bridgeless}).
A cactus $G$ is called a \emph{triangulated cactus} if all the cycles are of length three
and each edge belongs to a cycle (see \Cref{fig:triangulated}).
A cactus graph $G$ is called a \emph{core cactus} if there are no bridges that
split the graph such that each of the two components contain a cycle. In other words,
when all the bridges of a core cactus are removed, what remains is 
a bridge-less cactus (see \Cref{fig:corecactus}).
A graph $G$ is a \emph{bipartite cactus} or a \emph{bi-cactus} if $G$ is 
a cactus graph and also a bipartite graph (see \Cref{fig:bicactus}). 
Bi-pseudo-trees, bridge-less bi-cactus, and core bi-cactus are 
defined in a similar manner.


\begin{figure}
\centering
\begin{subfigure}[t]{0.3\textwidth}
\centering
\begin{tikzpicture}[scale=0.25]
\begin{scope}[
every node/.style={shape=circle,fill=blue}, 
every edge/.style={draw=black,ultra thick}
]

\node(A) at (-1,0) {};
\node(B) at (2,3) {};
\node(C) at (5,0) {};
\node(D) at (8,-1) {};
\node(E) at (9,-4) {};
\node(F) at (6,-6) {};
\node(G) at (3,-3) {};
\node(H) at (0,-2) {};
\node(I) at (-3,-2) {};
\node(J) at (-3,-5) {};
\node(K) at (0,-5) {};
\node(L) at (2,6) {};
\node(M) at (3,-6) {};
\node(N) at (12,-4) {};
\node(O) at (10,0) {};
\node(P) at (11,3) {};
\node(Q) at (13,2) {};
\node(S) at (12,5) {};
\node(R) at (15,5) {};
\node(T) at (9,5) {};
\node(U) at (11,7) {};
\node(V) at (10,9) {};
\node(W) at (7,9) {};
\node(X) at (6,7) {};
\node(Y) at (5,4) {};
\node(Z) at (7,2) {};

\path (A) [-] edge (B);
\path (A) [-] edge (C);
\path (B) [-] edge (C);
\path (C) [-] edge (D);
\path (D) [-] edge (E);
\path (E) [-] edge (F);
\path (F) [-] edge (G);
\path (C) [-] edge (G);
\path (G) [-] edge (H);
\path (H) [-] edge (I);
\path (I) [-] edge (J);
\path (J) [-] edge (K);
\path (H) [-] edge (K);
\path (B) [-] edge (L);
\path (K) [-] edge (M);
\path (E) [-] edge (N);
\path (D) [-] edge (O);
\path (O) [-] edge (P);
\path (O) [-] edge (Q);
\path (P) [-] edge (R);
\path (P) [-] edge (S);
\path (P) [-] edge (T);
\path (P) [-] edge (Z); %
\path (T) [-] edge (U);
\path (U) [-] edge (V);
\path (V) [-] edge (W);
\path (W) [-] edge (X);
\path (T) [-] edge (X);
\path (Y) [-] edge (Z);

\end{scope}
\end{tikzpicture}
\caption{A cactus graph.}
\label{fig:cactus}
\end{subfigure}
\hspace{10pt}
\begin{subfigure}[t]{0.3\textwidth}
\centering
\begin{tikzpicture}[scale=0.25]
\begin{scope}[
every node/.style={fill=blue,shape=circle}, 
every edge/.style={draw=black,ultra thick}
]
\node(A) at (-1,0) {};
\node(B) at (2,3) {};
\node(C) at (5,0) {};
\node(D) at (8,-1) {};
\node(E) at (9,-4) {};
\node(F) at (6,-6) {};
\node(G) at (5,-4) {};
\node(H) at (9,-6) {};
\node(J) at (-3,-4) {};
\node(K) at (1,-4) {};
\node(L) at (0,-2) {};
\node(P) at (11,3) {};
\node(Q) at (13,2) {};

\path (A) [-] edge (B);
\path (A) [-] edge (C);
\path (B) [-] edge (C);

\path (C) [-] edge (D);
\path (D) [-] edge (E);
\path (C) [-] edge (E);

\path (F) [-] edge (G);
\path (E) [-] edge (G);

\path (E) [-] edge (H);
\path (F) [-] edge (H);

\path (J) [-] edge (K);
\path (C) [-] edge (K);

\path (J) [-] edge (L);
\path (C) [-] edge (L);


\path (D) [-] edge (P);
\path (P) [-] edge (Q);
\path (D) [-] edge (Q);

\end{scope}

\end{tikzpicture}
\caption{A bridge-less cactus graph.}
\label{fig:bridgeless}
\end{subfigure}
\hspace{10pt}
\begin{subfigure}[t]{0.3\textwidth}
\centering

\begin{tikzpicture}[scale=0.25]

\begin{scope}[
every node/.style={fill=blue,shape=circle}, 
every edge/.style={draw=black,ultra thick}
]

\node(A) at (-1,0) {};
\node(B) at (2,3) {};
\node(C) at (5,0) {};
\node(D) at (8,-1) {};
\node(E) at (9,-4) {};
\node(F) at (6,-6) {};
\node(G) at (3,-3) {};
\node(H) at (11,-6) {};
\node(I) at (8,-6) {};
\node(J) at (-3,-7) {};
\node(K) at (0,-7) {};
\node(L) at (4,-9) {};
\node(M) at (1,-9) {};
\node(N) at (-2,-10) {};
\node(O) at (-4,-10) {};
\node(P) at (11,3) {};
\node(Q) at (13,2) {};

\path (A) [-] edge (B);
\path (A) [-] edge (C);
\path (B) [-] edge (C);

\path (C) [-] edge (D);
\path (D) [-] edge (E);
\path (C) [-] edge (E);

\path (E) [-] edge (F);
\path (F) [-] edge (G);
\path (E) [-] edge (G);

\path (E) [-] edge (H);
\path (H) [-] edge (I);
\path (E) [-] edge (I);

\path (G) [-] edge (J);
\path (J) [-] edge (K);
\path (G) [-] edge (K);

\path (K) [-] edge (L);
\path (L) [-] edge (M);
\path (K) [-] edge (M);

\path (K) [-] edge (N);
\path (N) [-] edge (O);
\path (K) [-] edge (O);

\path (D) [-] edge (P);
\path (P) [-] edge (Q);
\path (D) [-] edge (Q);

\end{scope}

\end{tikzpicture}
\caption{A triangulated cactus graph.}
\label{fig:triangulated}
\end{subfigure}

\vspace{20pt}

\begin{subfigure}[t]{0.3\textwidth}
\centering
\begin{tikzpicture}[scale=0.25] 
\begin{scope}[
every node/.style={shape=circle,fill=blue},
every edge/.style={draw=black,ultra thick}
]

\node(A) at (-1,0) {};
\node(B) at (2,2) {};
\node(C) at (5,0) {};
\node(D) at (8,-1) {};
\node(E) at (9,-4) {};
\node(F) at (6,-6) {};
\node(G) at (3,-3) {};
\node(H) at (0,-2) {}; %
\node(K) at (0,-5) {}; %
\node(L) at (2,5) {}; %
\node(M) at (3,-6) {};
\node(O) at (11,0) {}; %
\node(P) at (11,3) {}; %
\node(T) at (9,5) {}; %
\node(Y) at (5,4) {}; %
\node(Z) at (7,2) {}; %

\path (A) [-] edge (B);
\path (A) [-] edge (C);
\path (B) [-] edge (C);
\path (C) [-] edge (D);
\path (D) [-] edge (E);
\path (E) [-] edge (F);
\path (F) [-] edge (G);
\path (C) [-] edge (G);
\path[dashed] (G) [-] edge (H);
\path[dashed] (H) [-] edge (K);
\path[dashed] (B) [-] edge (L);
\path[dashed] (K) [-] edge (M);
\path[dashed] (D) [-] edge (O);
\path[dashed] (O) [-] edge (P);
\path[dashed] (P) [-] edge (T);
\path[dashed] (P) [-] edge (Z);
\path[dashed] (Y) [-] edge (Z);

\end{scope}
\end{tikzpicture}
\caption{A core cactus graph.}
\label{fig:corecactus}
\end{subfigure}
\hspace{10pt}
\begin{subfigure}[t]{0.3\textwidth}
\centering
\begin{tikzpicture}[scale=0.25] 
\begin{scope}[
every node/.style={shape=circle,fill=blue},
every edge/.style={draw=black,ultra thick}
]

\node(A) at (-1,0) {};
\node(B) at (2,3) {};
\node(C) at (5,0) {};
\node(D) at (8,-1) {};
\node(E) at (9,-4) {};
\node(F) at (6,-6) {};
\node(G) at (3,-3) {};
\node(H) at (0,-2) {};
\node(I) at (-3,-2) {};
\node(J) at (-3,-5) {};
\node(K) at (0,-5) {};
\node(L) at (5,4) {};
\node(M) at (3,-6) {};
\node(N) at (12,-3) {};

\path (A) [-] edge (B);
\path (A) [-] edge (C);
\path (C) [-] edge (L);
\path (C) [-] edge (D);
\path (E) [-] edge (F);
\path (F) [-] edge (G);
\path (C) [-] edge (G);
\path (G) [-] edge (H);
\path (H) [-] edge (I);
\path (I) [-] edge (J);
\path (J) [-] edge (K);
\path (H) [-] edge (K);
\path (B) [-] edge (L);
\path (K) [-] edge (M);
\path (E) [-] edge (N);
\path (D) [-] edge (N);
\end{scope}
\end{tikzpicture}
\caption{A bi-cactus graph.}
\label{fig:bicactus}
\end{subfigure}
\caption{Examples of a cactus graph, a bridge-less cactus graph, a triangulated cactus graph, 
a core cactus graph, and a bipartite cactus graph.
In the core cacti the dashed lines represent bridges and solid lines are the edges of the bridge-less core.}
\end{figure}
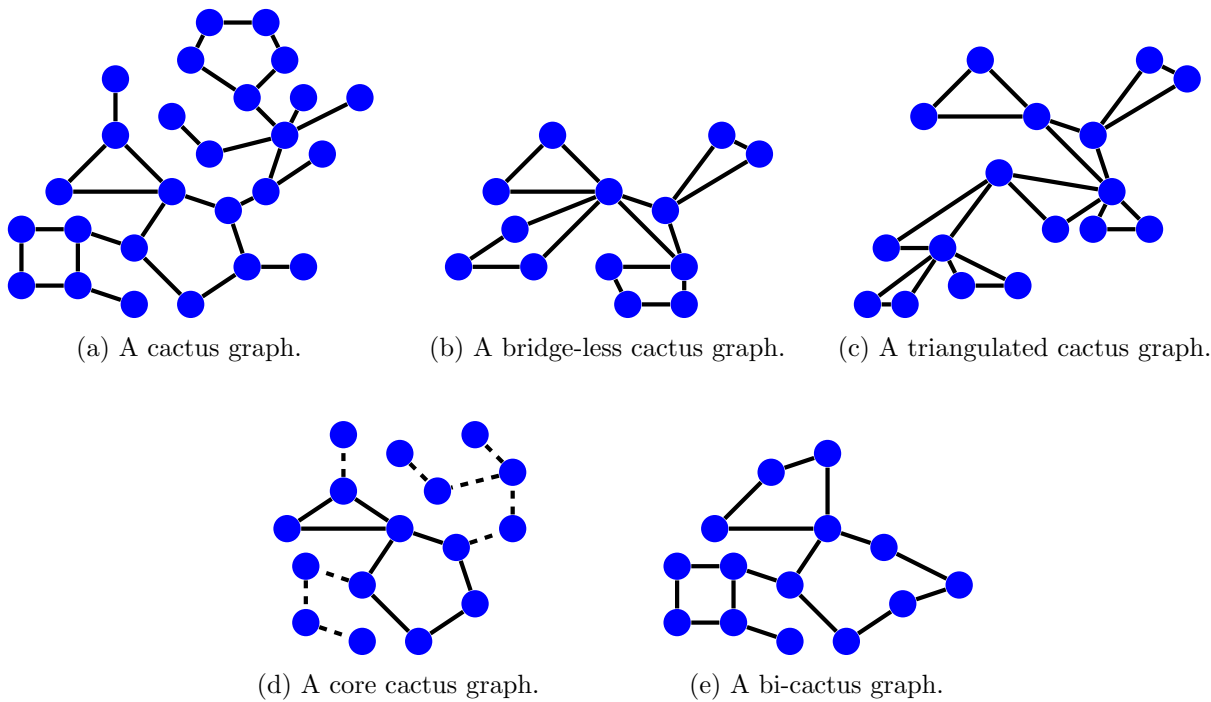


Given a connected graph $G$, the \emph{block-cutpoint graph} $\BC(G) = (V',E')$
of a graph $G$ is defined as follows~\cite{AgarssonGreenlaw06}.
Let $C(G) \subseteq V$ be the set of cut vertices, and let $\calB(G)$ be the set of blocks in $G$.
Then, 
\begin{align*}
V' & = C(G) \cup \calB(G) &
& \text{and} &
E' & = \set{(v,B) : v \in C(G), B \in \calB(G), v \in V(B)}
~.
\end{align*}
Observe that $\BC(G)$ is a tree.
See an example in \Cref{fig:blockcutpoint}.


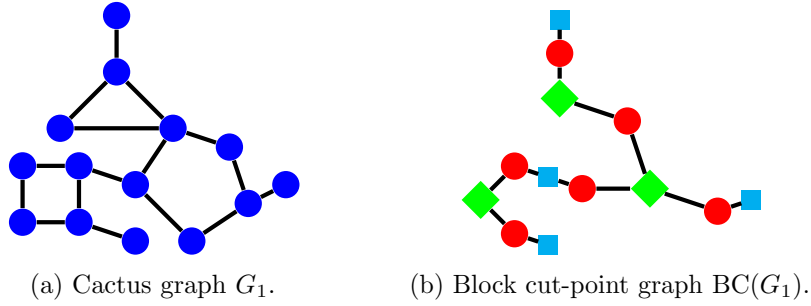
\begin{figure}[t]
\centering
\begin{subfigure}[t]{0.3\textwidth}
\centering
\begin{tikzpicture}[scale=0.25] 
\begin{scope}[
every node/.style={shape=circle,fill=blue},
every edge/.style={draw=black,ultra thick}
]

\node(A) at (-1,0) {};
\node(B) at (2,3) {};
\node(C) at (5,0) {};
\node(D) at (8,-1) {};
\node(E) at (9,-4) {};
\node(F) at (6,-6) {};
\node(G) at (3,-3) {};
\node(H) at (0,-2) {};
\node(I) at (-3,-2) {};
\node(J) at (-3,-5) {};
\node(K) at (0,-5) {};
\node(L) at (2,6) {};
\node(M) at (3,-6) {};
\node(N) at (11,-3) {};

\path (A) [-] edge (B);
\path (A) [-] edge (C);
\path (B) [-] edge (C);
\path (C) [-] edge (D);
\path (D) [-] edge (E);
\path (E) [-] edge (F);
\path (F) [-] edge (G);
\path (C) [-] edge (G);
\path (G) [-] edge (H);
\path (H) [-] edge (I);
\path (I) [-] edge (J);
\path (J) [-] edge (K);
\path (H) [-] edge (K);
\path (B) [-] edge (L);
\path (K) [-] edge (M);
\path (E) [-] edge (N);

\end{scope}
\end{tikzpicture}
\caption{Cactus graph $G_1$.}
\end{subfigure}
\hspace{20pt}
\begin{subfigure}[t]{0.32\textwidth}
\centering

\begin{tikzpicture}[scale=0.3] 

\begin{scope}[
every node/.style={fill=green,shape=circle},
every edge/.style={draw=black,ultra thick}
]

\node[fill=red](C1) at (2,3) {};
\node[fill=red](C2) at (5,0) {};
\node[fill=red](C6) at (9,-4) {};
\node[fill=red](C5) at (3,-3) {};
\node[fill=red](C3) at (0,-2) {};
\node[fill=red](C4) at (0,-5) {};

\node[shape=diamond](T) at (2,1) {};
\node[shape=diamond](S) at (-1.5,-3.5) {};
\node[shape=diamond](P) at (6,-3) {};

\node[fill=cyan,shape=rectangle](B1) at (2,4.5) {};
\node[fill=cyan,shape=rectangle](B2) at (1.5,-2.5) {};
\node[fill=cyan,shape=rectangle](B3) at (1.5,-5.5) {};
\node[fill=cyan,shape=rectangle](B4) at (10.5,-3.5) {};

\path (B1) [-] edge (C1);
\path (C1) [-] edge (T);
\path (C2) [-] edge (T);
\path (C3) [-] edge (S);
\path (C4) [-] edge (S);
\path (C3) [-] edge (B2);
\path (C5) [-] edge (B2);
\path (C4) [-] edge (B3);
\path (C2) [-] edge (P);
\path (C5) [-] edge (P);
\path (C6) [-] edge (P);
\path (C6) [-] edge (B4);
\end{scope}

\end{tikzpicture}
\caption{Block cut-point graph $\BC(G_1)$.}
\end{subfigure}
\caption{$G$ is a cactus graph; $\BC(G)$ is the block point-cut graph of $G$.
In $\BC(G)$ circles are cut-vertices, squares are bridge blocks, and diamonds are cycle blocks.
}
\label{fig:blockcutpoint}
\end{figure}


\section{Realization by Pseudo-Trees and Bi-Pseudo-Trees}
\label{sec:pseudotree}

In this section we give a characterization for degree realization by pseudo-trees and bi-pseudo-trees.
These results are used in the sequel, and serve as a warm-up.

\begin{observation}
If $G$ is a pseudo-forest, then $\sum_i d_i \in \set{2(n-c),\ldots,2n}$, 
where $c$ is the number of connected components in $G$.
If $G$ is a pseudo-tree, then $\sum_i d_i \in \set{2(n-1),2n}$.
\end{observation}

The realization problem is straightforward on trees and forests.

\begin{theorem}
\label{thm:forest}
{\bf (\cite{GJT07})}
Let $d$ be a degree sequence such that $\sum_i d_i \leq 2n - 2$.
Then $d$ has a forest realization with $(2n - \sum_i d_i)/2$ components.
If $\sum_i d_i = 2n - 2$, then $d$ has a tree realization.
\end{theorem}

The following observation considers the case, where $\sum_i d_i = 2n$ and $n \geq 3$.

\begin{observation}
Let $G$ be a pseudo-tree such that $n \geq 3$ and $\sum_i d_i = 2n$.
Then, $d_3 \geq 2$.
\end{observation}
\begin{proof}
If $d_3 = 1$, then $d_1 + d_2 = 2n - (n-2) = n + 2$. 
The sequence $d$ cannot be realized because $\multipl_1$ must be at least $n$ to satisfy the 
degree requirements $d_1$ and $d_2$ even if the vertices whose degrees are $d_1$ and $d_2$ 
are connected. A contradiction since $\multipl_1\le n-2$.
\end{proof}


\subsection{Unicyclic Realization}

We show that there is a realization by a unicyclic graph, if $\sum_i d_i = 2n$ and $d_3 \geq 2$. 
This was proven before in~\cite{BoeschHarary1978}. 
In this paper, we give a constructive proof that illustrates our approach for subsequent results.
More specifically, we use the minimum pivot version of the
Havel-Hakimi algorithm~\cite{havel55,Hakimi62} as long as the sequence contains a degree of 1.

\begin{theorem}
[\cite{BoeschHarary1978}]
\label{thm:pseudotree}
Let $d$ be a degree sequence such that $\sum_i d_i = 2n$ and $d_3 \geq 2$.
Then, the sequence $d$ has a unicyclic realization.
\end{theorem}
\begin{proof}
We prove the theorem by induction on $n - \multipl_2$.
The base case is a sequence $(2^n)$, for $n \geq 3$, for which
there is a realization of $d$ consisting of one cycle that contains all the vertices.
For the inductive step, assume that there is a unicyclic realization for sequences 
$d'$ such that $n' - \multipl'_2 < n - \multipl_2$.
Since $\sum_i d_i = 2n$ and $d \neq (2^n)$, it must be that $d_1 \geq 3$ and $d_n = 1$.
Moreover, $n > 3$, since $d_3 \geq 2$.
Let $d'$ be the sequence which is obtained by removing vertex $n$ and subtracting 1 from $d_1$.
Notice that $\sum_i d'_i = \sum_i d_i - 2 = 2n - 2 = 2n'$ and that $d'_3 \geq 2$.
Also, $n' - \multipl'_2  \leq n - \multipl_2 - 1$.
Hence, by the inductive hypothesis there is a unicyclic realization $G'$ of $d'$.
Obtain a realization $G$ of $d$ by adding the edge $(1,n)$.
(See \Cref{fig:leaf}.)
\end{proof}

\begin{figure}
\centering
\begin{tikzpicture}[scale=0.7]
\node [cloud, draw,cloud puffs=10,cloud puff arc=120, cloud ignores aspect, minimum width=4cm, minimum height=2cm, inner ysep=1em] {};

\begin{scope}[every node/.style={circle,fill=orange},
                       every edge/.style={draw=black,line width=1pt}]
\node (v1) at (1.5,0) {$d_1$};
\node (vn) at (4,0) {$1$};
\path (v1) [-] edge (vn);
\end{scope}
\end{tikzpicture}
\caption{Leaf addition step.}
\label{fig:leaf}
\end{figure}
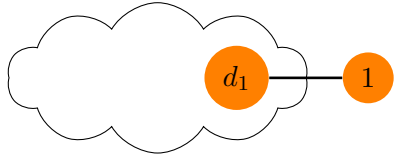

The above proof describes an algorithm that creates a cycle containing all vertices whose degree is larger than $1$.
Then, it adds degree-1 vertices as leaves to any vertex whose degree is greater than $2$.
Hence, \Cref{thm:forest} and the proof of \Cref{thm:pseudotree} imply the following.

\begin{corollary}
Let $d$ be a degree sequence such that $\sum_i d_i = 2n$ and $d_3 \geq 2$.
There is a linear time algorithm for computing a unicyclic realization of $d$
that contains a cycle of all vertices whose degree is greater than 1 (a.k.a. \emph{closed caterpillar}).
\end{corollary}


\subsection{Bi-Unicyclic Realization}

In the case of bi-unicyclic graph one needs to observe that there cannot be a realization if
$d = (2^n)$, where $n$ is odd.
In addition, if $d_4 = 1$, one cannot realize an even length cycle.
Hence, a realization algorithm should avoid such sequences.

\begin{theorem}
\label{thm:bipseudotree}
Let $d$ be a degree sequence such that $\sum_i d_i = 2n$.
The sequence $d$ has a bi-unicyclic realization if and only if
$d_4 \geq 2$ and $d \neq (2^n)$ for an odd $n$. 
\end{theorem}
\begin{proof}
If $d = (2^n)$, where $n$ is odd, then the only connected realization is a cycle.
There is no bi-unicyclic realization for $d$, since bipartite graphs cannot have odd cycles as sub-graphs.
The rest of the proof is similar to the proof of \Cref{thm:pseudotree}, and proceeds
by induction on $n - \multipl_2$.

There are two base cases.
The first case is when  $d = (2^n)$, where $n \geq 4$ and $n$ is even.
In this case, there is a realization of $d$ consisting of one even cycle that contains all the vertices.
The second case is when $d = (3,2^{n-2},1)$, and $n \geq 6$ is even.
In this case, $d$ can be realized by a cycle of length $n-2$, which is connected to a $P_2$, i.e., by 
a graph with the following edges set
\[
E = \set{ (i,i+1) : i = 1,\dots, n-3} \cup \set{(1,n-2), (1,n-1), (n-1,n)}
~.
\]
An example is given in See \Cref{fig:bipseudo} (for the case where $n = 6$).

For the inductive step, assume that the claim hold for sequences $d'$, 
such that $n' - \multipl'_2 < n - \multipl_2$.
Due to the induction base we have that $d \neq (2^n)$, for $n \geq 4$ and $n$ is even, 
and $d \neq (3,2^{n-2},1)$, and $n \geq 6$ is even.
Hence, it must be that either $n - \multipl_2 > 2$ or $d = (3,2^{n-2},1)$ and $n$ is odd.
In both cases, we have that $d_1 \geq 3$ and $d_n = 1$.
Moreover, $n \geq 5$, since $d_4 \geq 2$.
%
%
Let $d'$ be the sequence which is obtained by removing vertex $n$ and
subtracting 1 from $d_1$. Notice that $\sum_i d'_i = \sum_i d_i - 2 = 2n - 2 = 2n'$
and that $d'_4 \geq 2$. Also notice that $d' \neq (2^n)$, where $n$ is odd.
In addition, $n' - \multipl'_2 \leq n - \multipl_2 - 1$.
Hence, by the inductive hypothesis there is a bi-unicyclic realization $G'$ of $d'$.
We obtain a realization $G$ of $d$ by adding the edge $(1,n)$.
\end{proof}

\begin{figure}[t]
\centering
\begin{footnotesize}
\begin{tikzpicture}[scale=0.6]
\begin{scope}[every node/.style={fill=orange,shape=circle},
every edge/.style={draw=black,ultra thick}]
\node(A) at (0,0.5) {$2$};
\node(B) at (-2,2) {$2$};
\node(C) at (0,3.5) {$2$};
\node(D) at (2,2) {$3$};
\node[fill=red](F) at (4,2) {$2$};
\node(G) at (6,2) {$1$};
\path (A) [-] edge (B);
\path (B) [-] edge (C);
\path (C) [-] edge (D);
\path (D) [-] edge (A);
\path (D) [-] edge (F);
\path (F) [-] edge (G);
\end{scope}
\end{tikzpicture}
\end{footnotesize}
\caption{A bi-cactus realization of $d = (3,2^4,1)$.
Had the degree-1 vertex been attached to the degree-3 vertex, the residual sequence would have been $d' = (2^5)$.
The algorithm avoids this by placing a degree-2 vertex between the degree-1 vertex and the degree-3 vertex.}
\label{fig:bipseudo}
\end{figure}
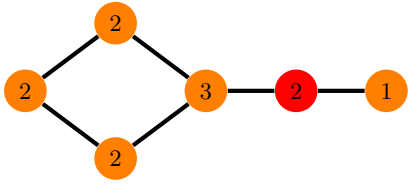

The following is implied by the proof of \Cref{thm:bipseudotree}.

\begin{corollary}
Let $d$ be a degree sequence such that $\sum_i d_i = 2n$, $d_4 \geq 2$, and $d \neq (2^n)$, where $n$ is odd.
There is a linear time algorithm that computes bi-unicyclic realization of $d$
that contains a cycle of all vertices whose degree is larger than 1,
maybe with the exception of one such vertex.
\end{corollary}


\section{Bounds on the Number of Edges} 
\label{sec:bounds}

In this section we provide upper bounds on the number of edges in cactus and
bi-cactus graphs.


\subsection{Bound on the Number of Edges in a Cactus}

Given a cactus graph $G$, let $c$ be the number of cycles in $G$
(not counting the outside face), let $t$ be the number of edges
in $G$ that belong to a cycle, and let $b$ be the number of bridges in $G$.
Notice that $m = b + t$.

The next observation is implied by the fact that a cactus graph is connected and planar.
More specifically, it is a direct implication of Euler's Formula.

\begin{observation}
\label{obs:euler}
Let $G$ be a cactus graph.
Then $m = n + c - 1$.
\end{observation}
\begin{proof}
Since $G$ is planar, given an embedding of $G$ in the plane,
Euler's formula implies that $m=n+f-2$, where $f$ is the number of faces.
As $G$ is a cactus graph, all faces in the embedding but the outside face are cycles, thus $f=c+1$.
\end{proof}

Next, we give an upper bound on the number of edges in a cactus graph.

\begin{lemma}
\label{lemma:cactusbound}
Let $G$ be a cactus graph.
Then, $m \leq \floor{\frac{3(n-1) - b}{2}}$.
\end{lemma}
\begin{proof}
Each edge is part of at most one cycle, so $c \leq (m-b)/3$.
\Cref{obs:euler} implies that
\[
m = n + c - 1 \leq n + (m-b)/3 - 1
~,
\]
or $2m  \leq 3n - 3 - b $.
\end{proof}

We now consider bridge-less cactus graphs and triangulated cactus graphs.

\begin{lemma}
\label{lemma:bridgelessbound}
Let $G$ be a bridge-less cactus graph. Then, $m \leq \floor{\frac{3(n-1)}{2}}$.
In particular, if $G$ is a triangulated cactus graph, then $n$ is odd and $m = 1.5(n-1)$.
\end{lemma}
\begin{proof}
The first bound is a direct implication of \Cref{lemma:cactusbound}.
Assume that $G$ is a triangulated cactus. Then $n$ is odd and $m = t = 3c$,
and thus by \Cref{lemma:cactusbound} we have that $m = 1.5(n-1)$.
\end{proof}

Let $G$ be a cactus graph and let $d=\deg(G)$ .
Recall that $\multipl_1$ is the number of $1$'s in $d$, and that $\multipl_{odd}$ is the number of odd integers that are larger than 1 in $d$.
Define the \emph{bridge parameter} of a sequence $d$ as follows:
\[
\beta
\triangleq \max \set{\multipl_1 , \half(\multipl_1 + \multipl_{odd})  }
~.
\]
Note that $\beta$ is an integer since $\multipl_1 + \multipl_{odd}$ is even.
For example, the cactus graph in \Cref{fig:cactus} has $\multipl_1 = 3$,
$\multipl_{odd} = \multipl_3 = 5$,
and $\beta = \half(\multipl_1 + \multipl_{odd}) = 4$.

We show that $\beta$ serves as a lower bound on the number of bridges in a cactus.

\begin{lemma}
\label{lemma:beta}
Let $G$ be a cactus graph, where $n > 2$. Then, $b \geq \beta$.
\end{lemma}
\begin{proof}
Any odd degree vertex must be connected to at least one bridge.
Hence, $b \geq \frac{\multipl_1 + \multipl_{odd}}{2}$.
In particular, the edge which is attached to a degree-1 vertex (a leaf) must be bridge,
and due to connectivity it must be connected to a vertex whose degree is greater than 1.
Thus, $b \geq \multipl_1$.
%
The lemma follows.
\end{proof}

\Cref{lemma:cactusbound,lemma:beta} imply the following
bound the number of edges in a cactus graph.
We note that this bound is implicit in~\cite{Rao81cactus}.

\begin{theorem}
\label{thm:beta}
Let $G$ be a cactus graph and $d=\deg(G)$.
Then
\(
m \leq \floor{\frac{3(n-1) - \beta}{2}}
\).
\end{theorem}


\subsection{Bound on the Number of Edges in a Bi-Cactus}

An obvious requirement from a bipartite graph is that all cycles have even length.

\begin{observation}
\label{obs:bicactus}
A cactus graph $G$ is bipartite if and only if all its cycles are of even length.
In particular, each cycle contains at least 4 edges.
Moreover, if $G$ is bridge-less, then $m = \half \sum_i d_i$ must be even.
\end{observation}

Another requirement is that the existence of a cycle requires at $n-\multipl_1 \geq 4$.

\begin{observation}
\label{obs:bicactus-d4}
Let $G$ be a bi-cactus graph such that $m \geq n$. Then $d_4 \geq 2$.
\end{observation}

The following lemma is the version of \Cref{lemma:cactusbound} for bi-cacti.
Its proof is somewhat more complicated.

\begin{lemma}
\label{lemma:bicactusbound}
Let $G$ be a bi-cactus graph, where $n \geq 4$.
Then, $m \leq 2\floor{\frac{2(n - 1 - b)}{3}} + b$.
\end{lemma}
\begin{proof}
Since $G$ is a bi-cactus each edge is part of at most one cycle, and
by \Cref{obs:bicactus} each cycle contains at least 4 edges.
It follows that $c \leq m/4$. However, one may obtain a tighter bound.
Consider the block cut-point tree $\BC(G)$, where the root is a cycle.
We remove blocks from $G$ according to $\BC(G)$ in a bottom up manner.
When one removes a bridge, both the number of edges and the number of
vertices in $G$ are reduced by $1$. When one removes a cycle of size $k$ from $G$, $k$ edges and $k-1$
vertices are removed, where $k$ is even and thus $k \geq 4$.
The highest ratio between the number of removed edges and the number of removed vertices
is obtained when $k = 4$, i.e., a ratio of $4/3$.
Assume that one is able to obtain this ratio of 4/3 for all cycle edges.
Then, the last cycle may be of size $4$, $6$ or $8$, depending on the remainder of dividing $n-b$ by 3.
Hence we get this ratio of 4 edges per 3 vertices from $n - b - k'$ vertices,
where $k'$ is the size of the last cycle.
Hence, the highest number of cycles is $(n - b - k')/3 + 1 = (n - b - k' + 3)/3$.

By~\Cref{obs:euler}, it follows that
\[
m = n + c - 1 \leq n + (n - b - k')/3
~.
\]
Let $n - 1 - b = 3q - r$, where $q = \ceil{\frac{n - 1 - b}{3}}$ and $r = 3q - (n - 1 - b)$.
Observe that $k' - 4 = 2r$.
Hence,
\begin{align*}
3m
& \leq~ 4(n - 1) - 2r - b \\
&~=~ 4(n-1) - 2 \paren{ 3\ceil{\frac{n - 1 - b}{3}} - (n - 1 - b) } - b \\
&=~ 6(n - 1 - b) - 6 \ceil{\frac{n - 1 - b}{3}} + 3b \\
&~=~ 6 \floor{\frac{2(n - 1 - b)}{3}} + 3b
~,
\end{align*}
where the last equality is due to $x = \floor{\frac{2x}{3}} + \ceil{\frac{x}{3}}$.
The lemma follows.
\end{proof}

In a bridge-less bi-cactus $b = 0$, and thus we obtain the following lemma.

\begin{lemma}
\label{lemma:bipartitebridgelessbound}
Let $G$ be a bridge-less bi-cactus,
where $n \geq 4$.
Then, $m \leq 2\floor{\frac{2(n - 1)}{3}}$.
\end{lemma}

The next example shows that one cannot replace $b$ with $\beta$
in the bound of \Cref{lemma:bicactusbound} as was done in the cactus case
(see \Cref{thm:beta}).

\begin{example}
\label{example:correction}
Consider the sequence $d = (4,3,2^6,1)$.
If we replace $b$ with $\beta$ in the bound of \Cref{lemma:bicactusbound},
we get an upper bound of
\[
m
\leq 2\floor{\frac{2(n - 1 - \beta)}{3}} + \beta
=    2\floor{\frac{2(9 - 1 - 1)}{3}} + 1
=    9
~.
\]
However, $d$ can be realized using $10$ edges as depicted in \Cref{fig:problem}.
Notice that there is an even degree vertex which is adjacent to two bridges.
In the sequel we show that one such \enquote{correction} for changing $b$ to $\beta$ in the bound of \Cref{lemma:bicactusbound} is enough.
\end{example}

\begin{figure}[t]
\centering
\begin{footnotesize}
\begin{tikzpicture}[scale=0.5]
\begin{scope}[every node/.style={fill=orange,shape=circle},
every edge/.style={draw=black,ultra thick}]

\node(A) at (0,0.5) {$2$};
\node(B) at (-2,2) {$2$};
\node(C) at (0,3.5) {$2$};
\node(D) at (2,2) {$4$};
\node(E) at (4,0.5) {$2$};
\node(F) at (6,2) {$3$};
\node(G) at (4,3.5) {$2$};
\node[fill=red](I) at (8,2) {$2$};
\node(J) at (10,2) {$1$};

\path (A) [-] edge (B);
\path (B) [-] edge (C);
\path (C) [-] edge (D);
\path (D) [-] edge (A);
\path (D) [-] edge (E);
\path (E) [-] edge (F);
\path (F) [-] edge (G);
\path (G) [-] edge (D);
\path (F) [-] edge (I);
\path (I) [-] edge (J);
\end{scope}
\end{tikzpicture}
\end{footnotesize}
\caption{A bi-cactus realization of $d = (4,3,2^6,1)$.}
\label{fig:problem}
\end{figure}
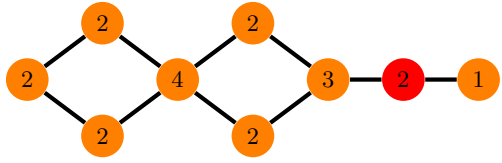

The following two technical lemmas are required for obtaining a 
bound on the number of edges in bi-cactus graphs. 

\begin{lemma}
\label{lemma:bicactus-technical}
$\floor{\frac{4(n-1) - \beta}{3}}
= \max \set{2\floor{\frac{2(n - 1 - \beta)}{3}} + \beta , 2\floor{\frac{2(n - 1 - (\beta + 1))}{3}} + (\beta +1) }$
\end{lemma}
\begin{proof}
Let $n - 1 - \beta = 3q - r$, where $q = \ceil{\frac{n - 1 - \beta}{3}}$ and $r = 3q - (n - 1 - \beta)$.
Observe that $r\in \{0,1,2\}$ by definition.
We have that
\[
2\floor{\frac{2(n - 1 - \beta)}{3}} + \beta
= 2\floor{\frac{2 (3q - r)}{3}} + \beta
= 4q + 2\floor{\frac{-2r}{3}} + \beta
= 4q - 2r + \beta
~,
\]
while
\begin{align*}
2\floor{\frac{2(n - 1 - (\beta+1))}{3}} + (\beta + 1)
& = 2\floor{\frac{2(3q - r - 1)}{3}} + \beta + 1 \\
& = 4q + 2\floor{\frac{-2r - 2}{3}} + \beta + 1 \\
& =
\begin{cases}
4q - 1 + \beta & r = 0 ~, \\
4q - 3 + \beta & r = 1, 2 ~.
\end{cases}
\end{align*}

If $r=0$, then the maximum is
\[
4q + \beta
= \frac{4}{3} (n - 1 - \beta) + \beta
= \frac{4(n-1) - \beta}{3}
= \floor{\frac{4(n-1) - \beta}{3}}
~.
\]

If $r=1$, then the maximum is
\[
4q + \beta - 2
= \frac{4}{3} (n - 1 - \beta + 1) + \beta - 2
= \frac{4(n-1) - \beta - 2}{3}
= \floor{\frac{4(n-1) - \beta}{3}}
~.
\]

If $r=2$, then the maximum is
\[
4q + \beta - 3
= \frac{4}{3} (n - 1 - \beta + 2) + \beta - 3
= \frac{4(n-1) - \beta - 1}{3}
= \floor{\frac{4(n-1) - \beta}{3}}
~. \qedhere
\]
\end{proof}

\begin{lemma}
\label{lemma:bicactus-technical2}
\(
2\floor{\frac{2(n - 1 - \beta)}{3}} + \beta  \geq 2\floor{\frac{2(n - 1 - (\beta + 2))}{3}} + (\beta +2)
\).
\end{lemma}
\begin{proof}
Observe that
\[
\frac{2(n - 1 - \beta)}{3}
= \frac{2(n - 1 - (\beta + 2))}{3} + \frac{4}{3}
~.
\]
Hence,
\[
\floor{\frac{2(n - 1 - \beta)}{3}}
\geq \floor{\frac{2(n - 1 - (\beta + 2))}{3}} + 1
~. \qedhere
\]
\end{proof}

\begin{theorem}
\label{thm:bicactusbeta}
A bi-cactus graph $G$ with $n \geq 4$ and $\beta \geq 1$ satisfies
$m \leq \floor{\frac{4(n - 1) - \beta}{3}}$.
\end{theorem}
\begin{proof}
\Cref{lemma:bicactusbound} provides an upper bound for bi-cactus graphs.
Also recall that $b \geq \beta$ by \Cref{lemma:beta}.
\Cref{lemma:bicactus-technical,lemma:bicactus-technical2}
imply that the bound is maximized either when $b = \beta$ or when $b = \beta+1$.
\end{proof}

Recall that a sequence $d$, such that  $m=n$, has no bi-unicyclic realization if $d_4 = 1$.
Observe that in this case $\beta = \multipl_1 \geq n-3$. Hence, the upper bound of \Cref{thm:bicactusbeta}
translates into
\[
m \leq \floor{\frac{4(n-1) - \multipl_1}{3}} \leq \floor{\frac{4(n-1) - (n-3)}{3}} = \floor{n - \frac{1}{3}} = n - 1
~,
\]
which means that there is no realization.


\section{Realization of Bridge-less Cactus and Bi-Cactus Graphs}
\label{sec:bridge-less}

In this section we consider bridge-less cactus graphs and their bipartite version.


\subsection{Bridge-less Cactus Graph Realization}

We give a characterization and a realization algorithm for bridge-less cacti.
We first prove that a bridge-less cactus is a cactus with even degrees and vice versa.

\begin{lemma}
\label{lemma:bridgeless}
A cactus graph $G$ is bridge-less if and only if $d_i$ is even, for every $i$.
\end{lemma}
\begin{proof}
Suppose that $G$ has no bridges.
Consider a vertex $v$. Each cycle that contains $v$ contributes exactly $2$ to its degree.
Hence, $v$'s degree is even.

Suppose that $d_i$ is even, for every $i$.
Assume that $G$ contains a bridge $(x,y)$.
Since $\deg(x)$ is even, it must be adjacent to another bridge $(x,z)$.
Consider the block-cutpoint graph $\BC(G)$ of $G$. Recall that $\BC(G)$ is a tree.
A bridge node cannot be a leaf of $\BC(G)$, since this means that there must be a vertex of degree 1 in $G$.
Hence, all leaves of $\BC(G)$ are cycle nodes.
There must be a bridge node whose removal splits $\BC(G)$ into two trees, one of which does
not contains bridge nodes. Let this bridge be $(x,y)$ in $G$. It follows that either $x$ or $y$ have
an odd degree. A contradiction.
\end{proof}

Next we show that a degree sequence has a realization as a bridge-less cactus
if and only if it satisfies the bound of \Cref{lemma:bridgelessbound} and it consists of even numbers.

\begin{theorem}
\label{thm:bridgeless}
Let $d$ be a degree sequence of length $n \geq 3$.
There is a bridge-less cactus realization of $d$ if and only if
$m \leq \floor{1.5(n-1)}$ and $d_i$ is even, for every $i$.
\end{theorem}
\begin{proof}
If there is a bridge-less cactus realization, then \Cref{lemma:bridgeless} and
\Cref{lemma:bridgelessbound} imply that $m \leq \floor{1.5(n-1)}$ and
$d_i$ is even, for every $i$.

The converse is proved by induction on $m-\multipl_2$.
In the base case $d = (2^n)$, where $n \geq 3$, there is a realization of $d$ consisting of one cycle that contains all vertices.
For the inductive step, since $d \neq (2^n)$, it must be that $d_1 \geq 4$ since $d_1$ is even.
Moreover, it must be that $n \geq 5$, since $m \geq 2 + 1(n-1) = n + 1 > \floor{1.5(n-1)}$,
for $n \leq 4$.
Also, since $\sum_i d_i < 3n$, there must be more than $n/2$ vertices of degree 2 in $d$.
In particular, $d_n = d_{n-1} = 2$.
Let $d'$ be the sequence which is obtained by removing $n$ and $n-1$ and subtracting 2 from $d_1$.
Notice that $n' \geq 3$ because $n'=n-2$. Also, since $2m=\sum_i d_i$, we have that
\[
\sum_i d'_i
=    \sum_i d_i - 6
\leq 2\floor{1.5(n-1)} - 6
=    2\floor{1.5(n'-1)}
~.
\]
In addition, $m' - \multipl'_2 \leq m - 3 - (\multipl_2 - 2) = m - \multipl_2 - 1$.
By the induction hypothesis $d'$ has a realization as a bridge-less cactus $G'$.
We obtain a realization $G$ for $d$ by adding a triangle of the vertices $1$, $n-1$, and $n$.
(See \Cref{fig:triangle}.)
\end{proof}

\begin{figure}
\centering
\begin{tikzpicture}[scale=0.7]
\node [cloud, draw,cloud puffs=10,cloud puff arc=120, cloud ignores aspect, minimum width=4cm, minimum height=2cm, inner ysep=1em] {};

\begin{scope}[every node/.style={circle,fill=orange},
                       every edge/.style={draw=black,line width=1pt}]
\node (v1) at (1.5,0) {$d_1$};
\node (vn) at (4,1) {$2$};
\node (vn1) at (4,-1) {$2$};
\path (v1) [-] edge (vn1);
\path (v1) [-] edge (vn);
\path (vn1) [-] edge (vn);
\end{scope}
\end{tikzpicture}
\caption{Triangle addition step.}
\label{fig:triangle}
\end{figure}
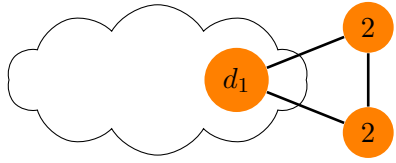

A similar approach also works for triangulated cacti.
This result already appeared in~\cite{Rao81cactus}.
The proof is given for completeness.

\begin{theorem}[\cite{Rao81cactus}]
\label{thm:triangulated}
Let $d$ be a degree sequence of length $n \geq 3$.
There is a triangulated cactus realization of $d$ if and only if
$n$ is odd, $m = 1.5(n-1)$, and $d_i$ is even, for every $i$.
\end{theorem}
\begin{proof}
If there is a realization of $d$ as a triangulated cactus,
then \Cref{lemma:bridgeless,lemma:bridgelessbound} imply that
$n$ is odd, $m = 1.5(n-1)$ and $d_i$ is even, for every $i$.

The converse is proved by induction on $n$.
In the base case $n = 3$, and we have that $\sum_i d_i = 6(3-1) = 6$,
which means that $d = (2^3)$. Hence, the only realization is a triangle.
For the inductive step, since $n$ is odd, we have that $n \geq 5$.
Since $\sum_i d_i = 3(n-1)$ and $d_1$ is even, it must be that $d_1 \geq 4$.
Also, there are more than $n/2$ vertices of degree 2 in $d$. In particular, $d_n = d_{n-1} = 2$.
Let $d'$ be the sequence obtained by removing $n$ and $n-1$ and subtracting 2 from $d_1$.
Notice that $n' \geq 3$ and $n'$ is odd. Also,
\[
\sum_i d'_i
= \sum_i d_i - 6
= 3(n-1) - 6
= 3(n'-1)
~.
\]
By the induction hypothesis, $d'$ has a realization as a triangulated cactus $G'$.
We obtain a realization $G$ for $d$ by adding a triangle of the vertices $1$, $n-1$, and $n$.
\end{proof}

The proofs of \Cref{thm:bridgeless,thm:triangulated} imply an algorithm that repeatedly
forms a triangle composed of two vertices whose current degree is $2$ and of one vertex whose
current degree is at least $4$ until only degree-2 vertices remain. Then, a cycle is created of all
remaining vertices.

\begin{corollary}
Let $d$ be a degree sequence such that $n \geq 3$, $m \leq \floor{1.5(n-1)}$,
and $d_i$ is even, for every $i$.
There is a linear time algorithm that computes a bridge-less cactus realization of $d$,
where all cycles except maybe one are triangles.
\end{corollary}


\subsection{Bridge-less Bi-Cactus Graph Realization}

We provide a characterization and a realization algorithm
for bridge-less bi-cactus graphs. The approach is similar to the one for
bridge-less cactus graphs, where the main difference is that we use cycles
of length 4 and not triangles.

\begin{theorem}
\label{thm:bipartite-bridgeless}
Let $d$ be a degree sequence of length $n \geq 4$.
There is a realization of $d$ as a bridge-less bi-cactus if and only if
$m \leq 2\floor{\frac{2(n-1)}{3}}$, $m$ is even, and $d_i$ is even, for every $i$.
\end{theorem}
\begin{proof}
If there is a realization, then \Cref{lemma:bridgeless}, \Cref{obs:bicactus}, and
\Cref{lemma:bipartitebridgelessbound} imply that $d_i$ is even, for every $i$,
$m$ is even, and $m \leq 2\floor{\frac{2(n-1)}{3}}$.

The converse is proved by induction on $m-\multipl_2$.
In the base case $d = (2^n)$, where $n \geq 4$. Since $m$ is even, $n$ is also even.
There is a realization of $d$ consisting of one cycle that contains all vertices.

For the inductive step, $d \neq (2^n)$, implies that $d_1 \geq 4$.
It must be that $n \geq 7$, since otherwise
\[
\sum_i d_i \geq 4 + 2(n-1) = 2(n + 1) > 4\floor{\frac{2(n - 1)}{3}}
~.
\]
Also, there must be more than $2n/3$ vertices of degree 2 in $d$.
In particular, $d_{n-2} = 2$.
Let $d'$ be the sequence obtained by removing $n-2$, $n-1$, and $n$
and subtracting 2 from $d_1$. Notice that $n' = n - 3 \geq 4$ and $m' = m - 4$,
which means that $m$ is even. Also,
\[
\sum_i d'_i
=    \sum_i d_i - 8
\leq 4\floor{\frac{2(n - 1)}{3}} - 8
=    4\floor{\frac{2(n - 1)}{3} - 2}
=    4\floor{\frac{2(n' - 1)}{3}}
~.
\]
In addition, $m' - \multipl'_2 \leq m - 4 - (\multipl_2 - 3) = m - \multipl_2 - 1$.
By the inductive hypothesis $d'$ has a realization as a bridge-less bi-cactus $G'$.
We obtain a realization $G$ for $d$ by adding the edges $(1,n-2)$, $(n-2,n-1)$, $(n-1,n)$,
and $(1,n)$. (See \Cref{fig:C4}.)
\end{proof}

\begin{figure}
\centering
\begin{tikzpicture}[scale=0.7]
\node [cloud, draw,cloud puffs=10,cloud puff arc=120, cloud ignores aspect, minimum width=4cm, minimum height=2cm, inner ysep=1em] {};

\begin{scope}[every node/.style={circle,fill=orange},
                       every edge/.style={draw=black,line width=1pt}]
\node (v1) at (1.5,0) {$d_1$};
\node (vn) at (4,-1) {$2$};
\node (vn1) at (6,0) {$2$};
\node (vn2) at (4,1) {$2$};
\path (v1) [-] edge (vn2);
\path (v1) [-] edge (vn);
\path (vn2) [-] edge (vn1);
\path (vn1) [-] edge (vn);
\end{scope}
\end{tikzpicture}
\caption{$C_4$ addition step.}
\label{fig:C4}
\end{figure}
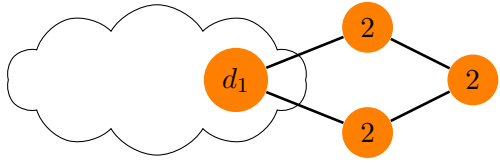

\begin{corollary}
\label{cor:bi-bridgeless}
A degree sequence $d$, where $n \geq 4$, $m \leq 2\floor{\frac{2(n-1)}{3}}$,
$m$ is even, and $d_i$ is even for every $i$, 
has a linear time algorithm that computes a bridge-less bi-cactus realization of $d$
where all cycles except maybe one are of length $4$.
\end{corollary}


\section{Realization of Core Cactus and Bi-Cactus Graphs}
\label{sec:core}

In this section we consider core cactus graphs and core bi-cactus graphs.


\subsection{Core Cactus Graph Realization}

We provide a characterization and a realization algorithm for core cactus graphs.
As a first step, we observe that in core cacti we have that $\beta = \multipl_1$.

\begin{lemma}
\label{lemma:core}
If $G$ is a core cactus, then $\multipl_1 \geq \multipl_{odd}$.
\end{lemma}
\begin{proof}
We prove the lemma by induction on the block point-cut graph $\BC(G)$.
The base case is a bridge-less cactus, in which $\multipl_1 = \multipl_{odd} = 0$.
For the inductive step, we remove a bridge leaf from $\BC(G)$.
This amounts to removing a degree-1 vertex denoted $j$.
This also lowers the degree of the vertex $k$ on the other side of this bridge.
Observe that $d_k \geq 2$.
Let $G'$ be the graph without $j$. By the inductive hypothesis we have that
$\multipl'_1 \geq \multipl'_{odd}$.
There are several options depending on $d_k$:
\begin{itemize}
\item If $d_k = 2$, then $\multipl_1 = \multipl'_1 \geq \multipl'_{odd} = \multipl_{odd}$.
\item If $d_k \geq 4$ is even, then $\multipl_1 = \multipl'_1 + 1 \geq \multipl'_{odd} + 1 = \multipl_{odd} + 2$.
\item If $d_k \geq 3$ is odd , then $\multipl_1 = \multipl'_1 + 1 \geq \multipl'_{odd} + 1 = \multipl_{odd}$.
\qedhere
\end{itemize}
\end{proof}

We show that a degree sequence $d$ has a realization as a core cactus if and only if
it satisfies the upper bound of \Cref{thm:beta} and $\beta = \multipl_1$.

\begin{theorem}
\label{thm:corecactus}
Let $d$ be a degree sequence such that $n \geq 3$ and $\sum_i d_i > 2n$.
There is a core cactus realization of $d$ if and only if $\multipl_{odd} \leq \multipl_1$
and $m \leq \floor{\frac{3(n - 1) - \multipl_1}{2}}$.
\end{theorem}
\begin{proof}
If there is a core cactus realization $G$ of $d$, then $\multipl_1 \geq \multipl_{odd}$ by \Cref{lemma:core}.
Observe that in this case $\beta = \multipl_1$ by definition. Hence, \Cref{thm:beta} implies that
\[
m
\leq \floor{\frac{3(n-1) - \beta}{2}}
=    \floor{\frac{3(n-1) - \multipl_1}{2}}
~.
\]
The converse is proved by induction on $\multipl_1 + \multipl_{odd}$.
In the base case, $\multipl_1 = \multipl_{odd} = 0$, and thus $d_i$ is even for every $n$.
Since $m \leq \floor{1.5(n-1)}$, it follows by \Cref{thm:bridgeless} that
there exists a realization for $d$ as a bridge-less cactus.

For the inductive step, there are two cases.
If $d$ contains an odd number $d_j > 1$, then since $\multipl_{odd} \leq \multipl_1$,
it must be that $d_n = 1$.
Let $d'$ be the sequence 
obtained by subtracting 1 from $d_j$ and removing $n$.
Observe that $\multipl'_{odd} = \multipl_{odd} - 1 \leq \multipl_1 - 1 = \multipl'_1$.
In addition.
\[
\sum_i d'_i = \sum_i d_i - 2 > 2n - 2 = 2n'
~,
\]
and
\[
\sum_i d'_i
=    \sum_i d_i - 2
\leq 2\floor{\frac{3(n - 1) - \multipl_1}{2}} - 2
=    2\floor{\frac{3(n - 1) - \multipl_1 - 2}{2}}
=    2\floor{\frac{3(n' - 1) - \multipl'_1}{2}}
~.
\]
By the inductive hypothesis there is a realization $G'$ of $d'$ as a core cactus.
We obtain a core cactus realization $G$ of $d$ by adding the edge $(j,n)$.

Suppose that $d$ does not contain an odd number $d_j > 1$, but $d_n = 1$.
Then, it must be the case that $d_{n-1} = 1$ and $d_1 \geq 4$.
Let $d'$ be the sequence which is obtained by subtracting 2 from $d_1$ and removing $n-1$ and $n$.
Observe that $0 = \multipl'_{odd} \leq \multipl'_1 = \multipl_1 - 2$.
Also,
\[
\sum_i d'_i = \sum_i d_i - 4 > 2n - 4 = 2n'
~,
\]
and
\[
\sum_i d'_i
=    \sum_i d_i - 4
\leq 2\floor{\frac{3(n-1) - \multipl_1}{2}} - 4
=    2\floor{\frac{3(n-1) - \multipl_1 - 4}{2}}
=    2\floor{\frac{3(n'-1) - \multipl'_1}{2}}
~.
\]
By the induction hypothesis there is a core cactus realization $G'$ of $d'$.
We obtain a core cactus realization $G$ of $d$ by adding the edges $(1,n-1)$ and $(1,n)$.
\end{proof}

The algorithm which is implied by the proof of \Cref{thm:corecactus} initially
connects 1-degree vertices to vertices with odd degree which is greater than 1.
When $\multipl_{odd} = 0$, it attaches two degree-1 vertices to a vertex with
even degree which is larger than $2$.
When all degrees are even, it constructs a bridge-less graph
(the core).

\begin{corollary}
A degree sequence $d$ where
$n \geq 3$, $n < m \leq \floor{\frac{3(n - 1) - \multipl_1}{2}}$ and $\multipl_{odd} \leq \multipl_1$
has a linear time algorithm that computes a core cactus realization of $d$,
where all cycles except maybe one are triangles, and
all other edges are connected to cycle vertices.
\end{corollary}


\subsection{Core Bi-Cactus Graph Realization}

In this section we provide a characterization and a realization algorithm
for core bi-cacti. The approach is similar to the one for
core cactus graphs. One difference is that we use cycles of length 4 and not triangles.
Another is that we sometimes need to use a correction as shown in \Cref{fig:problem}.

We need the following technical lemma.

\begin{lemma}
\label{lemma:technical3}
$\floor{ \half \floor{\frac{4(n - 1) + 1}{3}} } =  \floor{\frac{2(n - 1)}{3}}$.
\end{lemma}
\begin{proof}
Let $n - 1 = 3q - r$, where $q = \ceil{\frac{n - 1}{3}}$ and $r = 3q - (n - 1)$.

If $r = 0$, then
\[
\floor{ \half \floor{\frac{4(n - 1) + 1}{3}} }
= \floor{ \half \floor{\frac{12q + 1}{3}} }
= 2q
= \floor{\frac{2 \cdot 3q}{3}}
= \floor{\frac{2(n - 1)}{3}}
~.
\]

If $r = 1$, then
\[
\floor{ \half \floor{\frac{4(n - 1) + 1}{3}} }
= \floor{ \half \floor{\frac{12q - 4 + 1}{3}} }
= 2q - 1
= \floor{\frac{2(3q - 1)}{3}}
= \floor{\frac{2(n - 1)}{3}}
~.
\]

If $r = 2$, then
\[
\floor{ \half \floor{\frac{4(n - 1) + 1}{3}} }
= \floor{ \half \floor{\frac{12q - 8 + 1}{3}} }
= 2q - 2
= \floor{\frac{2(3q - 2)}{3}}
= \floor{\frac{2(n - 1)}{3}}
~.
\qedhere
\]
\end{proof}

\begin{theorem}
\label{thm:corebicactus}
A degree sequence $d$ where $d_4 \geq 2$, $\sum_i d_i > 2n$ and $\multipl_1 > 0$
has a core bi-cactus realization if and only if
$\multipl_1 \geq \multipl_{odd}$ and $m \leq \floor{\frac{4(n - 1) - \multipl_1}{3}}$. 
\end{theorem}
\begin{proof}
If there is a core bi-cactus realization of $d$, then by \Cref{lemma:core}
we have that $\multipl_1 \geq \multipl_{odd}$. Moreover,
\Cref{thm:bicactusbeta} implies that
\[
m
\leq \floor{\frac{4(n - 1) - \beta}{3}}
=    \floor{\frac{4(n - 1) - \multipl_1}{3}}
~.
\]

The converse is proved by induction on $\multipl_1 + \multipl_{odd}$.
In the base case, $\multipl_1 + \multipl_{odd} = 0$. In this case,
$d_i$ is even for every $i$, $n \geq 4$, and $m$ is even.
Since $m$ is even, we have that $m\leq 2\floor{2(n-1)/3}$, and it follows by \Cref{thm:bipartite-bridgeless} 
that there is a realization for $d$ as a bridge-less bi-cactus.

For the inductive step, there are two cases.
If $d$ contains an odd number $d_j > 1$, then since $\multipl_{odd} \leq \multipl_1$,
it must be that $d_n = 1$.
Let $d'$ be the sequence obtained by subtracting 1 from $d_j$ and removing $n$.
Observe that $\multipl'_{odd} = \multipl_{odd} - 1 \leq \multipl_1 - 1 = \multipl'_1$,
$d'_4 \geq 2$,
and
\[
\sum_i d'_i = \sum_i d_i - 2 > 2n - 2 = 2n'
~.
\]
In addition,
\[
\sum_i d'_i
=    \sum_i d_i - 2
\leq 2\floor{\frac{4(n - 1) - \multipl_1}{3}} - 2
=    2\floor{\frac{4(n' - 1) - \multipl'_1}{3}}
~.
\]
First, suppose that $\multipl'_1 \geq 1$ or $m'$ is even.
By the inductive hypothesis there is a realization $G'$ of $d'$ as a core bi-cactus.
We obtain a core cactus realization $G$ of $d$ by adding the edge $(j,n)$.
If $\multipl'_1 = 0$ and $m'$ is odd, we create another sequence $d^*$ by removing $d'_{n'} = 2$.
Observe that $n^* = n' - 1$, $m^* = m' - 1$, and $\multipl'_1 = 0$.
Hence,
\[
\sum_i d^*_i  = \sum_i d'_i - 2 > 2n' - 2 = 2n^*
~,
\]
and
\[
m^*
=    m' - 1
\leq \floor{\frac{4(n' - 1)}{3}} - 1
=    \floor{\frac{4(n^* - 1) + 1}{3}}
~.
\]
Since $m^*$ is even, we have that
\[
m^*
\leq 2\floor{ \half \floor{\frac{4(n^* - 1) + 1}{3}} }
~.
\]
By \Cref{lemma:technical3} it follows that
\(
m^*
\leq 2\floor{2(n^* - 1)/3}
\).
Moreover, since $m > n$, it must be that $n \geq 8$. Hence, $n^* \geq 6$.
By the inductive hypothesis there is a realization $G^*$ of $d^*$ as a core bi-cactus.
We obtain a core bi-cactus realization $G$ of $d$ by adding the edges $(j,n-1)$ and $(n-1,n)$.

Suppose that $d$ does not contain an odd number $d_j > 1$, but $d_n = 1$.
Then, it must be that $d_{n-1} = 1$ and $d_1 \geq 4$.
Let $d'$ be the sequence obtained by subtracting 2 from $d_1$ and removing $n-1$ and $n$.
Observe that $0 = \multipl'_{odd} \leq \multipl'_1 = \multipl_1 - 2$, $d'_4 \geq 2$, and
\[
\sum_i d'_i = \sum_i d_i - 4 > 2n - 4 = 2n'
~.
\]
Also,
\[
\sum_i d'_i
=    \sum_i d_i - 4
\leq 2\floor{\frac{4(n - 1) - \multipl_1}{3}} - 4
=    2\floor{\frac{4(n' - 1) - \multipl'_1}{3}}
~.
\]
Suppose that $\multipl'_1 \geq 1$ or $m$ is even.
By the inductive hypothesis there is a realization $G'$ of $d'$ as a core cactus.
We obtain a core bi-cactus realization $G$ of $d$ by adding the edges $(1,n-1)$ and $(1,n)$.

If $\multipl'_1 = 0$ and $m$ is odd, continue as in the first case.
Since $n \geq 8$, we have that $n^* \geq 5$.
We obtain a core bi-cactus realization $G$ of $d$ by adding the edges $(1,n-2)$, $(n-2,n-1)$, and $(1,n)$.
\end{proof}

\begin{corollary}
Let $d$ be a degree sequence such that $d_4 \geq 2$, 
$\multipl_1 \geq \max\set{\multipl_{odd},1}$ and $n < m \leq \floor{\frac{4(n - 1) - \multipl_1}{3}}$.
Then, there is a linear time algorithm that computes a core bi-cactus realization of $d$,
where all cycles except maybe one are of length $4$.
Also, all other edges, but maybe one, are connected to cycle vertices.
\end{corollary}


\section{Realization of Cactus and Bi-Cactus Graphs}
\label{sec:cactus}

In this section we give a characterization for realization of both
cactus graphs and bi-cactus graphs.


\subsection{Cactus Graph Realization}

Characterization and a realization algorithm for cactus graphs were given before in~\cite{Rao81cactus}.
We include a proof of the following theorem for completeness.

\begin{theorem}[\cite{Rao81cactus}]
\label{thm:cactus}
Let $d$ be a degree sequence such that $n \geq 3$ and $\sum_i d_i \geq 2n$.
Then there is a cactus realization of $d$ if and only if $m \leq \floor{\frac{3(n - 1) - \beta}{2}}$.
\end{theorem}
\begin{proof}
First, if there is a realization of $d$ as a cactus, then \Cref{thm:beta} implies that
$m \leq \floor{\frac{3(n-1) - \beta}{2}}$.

For the other direction, suppose that
$m \leq \floor{\frac{3(n - 1) - \beta}{2}}$.
If $\multipl_{odd} \leq \multipl_1$, then \Cref{thm:corecactus} implies that there is
a realization of $d$ as a core cactus.
So now suppose $\multipl_{odd} > \multipl_1$.
We prove the claim by induction on $\multipl_1 + \multipl_{odd}$ and on $\multipl_{odd}$.
In the base case, there are two options.

\begin{compactitem}
\item If $\sum_i d_i = 2n$, then by \Cref{thm:pseudotree} there exists a unicyclic realization of $d$.

\item $d_i$ is even for every $n$ or $\multipl_1 = \multipl_{odd} = 0$.
        Since $\sum_i d_i \leq 2\floor{1.5(n-1)}$,
        it follows by \Cref{thm:bridgeless} that there exists a realization for $d$
        as a bridge-less cactus.
\end{compactitem}

For the inductive step, there are two cases.
%
%
First, suppose that $\multipl_1 > 0$.
Since $\multipl_{odd} \geq \multipl_1$, the sequence $d$ must contain an odd number $d_j \geq 3$.
Let $d'$ be the sequence which is obtained by subtracting 1 from $d_j$
and removing $n$.
Observe that $\multipl'_{odd} = \multipl_{odd} - 1 \geq \multipl_1 - 1 = \multipl'_1$,
\[
\sum_i d'_i = \sum_i d_i - 2 > 2n - 2 = 2n'
~,
\]
and that
\begin{align*}
\sum_i d'_i
=    \sum_i d_i - 2
& \leq 2\floor{ \frac{3(n - 1) - \half(\multipl_1 + \multipl_{odd})}{2} } - 2 \\
& =     2\floor{ \frac{3(n - 1) - \half(\multipl_1 + \multipl_{odd}) - 2}{2} }  
 =     2\floor{ \frac{3(n' -1) - \half(\multipl'_1 + \multipl'_{odd})}{2} }
~.
\end{align*}
By the induction hypothesis there is a realization $G'$ of $d'$ as a cactus.
We obtain a cactus realization $G$ of $d$ by adding the edge $(j,n)$.


Suppose that $\multipl_1 = 0$.
In this case, $d$ contains at least two odd numbers, i.e., $\multipl_{odd} \geq 2$.
Let $d_j \geq 3$ be the smallest odd number in $d$.
Since $\sum_i d_i \leq 2\floor{\frac{3(n - 1) - \beta}{2}}$, it must be that $d_n = d_{n-1} = 2$.
Let $d'$ be the sequence which is obtained by removing $n$ and $n-1$ and
subtracting 2 from $d_j$.
Observe that $\multipl'_{odd} + \multipl'_1 = \multipl_{odd} + \multipl_1$.
In particular, if $d_j=3$, then $\multipl'_1=1$, and otherwise $\multipl'_1=0$.
Also,
\[
\sum_i d'_i = \sum_i d_i - 6 > 2n - 6 = 2n' - 2
~,
\]
namely, $\sum_i d'_i \geq 2n'$.
In addition,
\begin{align*}
\sum_i d'_i
=    \sum_i d_i - 6
& \leq 2\floor{ \frac{3(n - 1) - \half(\multipl_1 + \multipl_{odd})}{2} } - 6 \\
& =    2\floor{ \frac{3(n - 1) - \half(\multipl_1 + \multipl_{odd}) - 6}{2} } 
 =    2\floor{ \frac{3(n' -1) - \half(\multipl'_1 + \multipl'_{odd})}{2} }
~.
\end{align*}
By the induction hypothesis $d'$ has a realization as a cactus $G'$.
We obtain a realization $G$ for $d$ by adding a triangle of the vertices $j$, $n-1$, and $n$.
\end{proof}

The algorithm which is implied by our proof of \Cref{thm:cactus} works as follows.
If $\multipl_1 \geq \multipl_{odd}$ it constructs a core cactus.
Otherwise, it connects 1-degree vertices to vertices with an odd degree which is greater
than 1. When $\multipl_1 = 0$, and as long as $\multipl_{odd} > 0$, it adds a triangle
consisting of two degree-2 vertices and a vertex $j$ with the smallest odd degree.
This is done until the degree of $j$ becomes $1$.
If the volume becomes $2n$, then a unicyclic graph is constructed.
Otherwise, a sequence consisting of even numbers is obtained,
and a bridge-less cactus is created.

\begin{corollary}
Let $d$ be a degree sequence such that $n \geq 3$, $n - 1 \leq m \leq \floor{\frac{3(n - 1) - \beta}{2}}$.
There is a linear time algorithm that computes a cactus realization of $d$,
where all cycles except maybe one are triangles.
\end{corollary}


\subsection{Bi-Cactus Graph Realization}

In this section we provide a full characterization and a realization algorithm
for bicactus graphs. The approach is similar to the one for cactus graphs.

\begin{theorem}
\label{thm:bicactus}
Let $d$ be a degree sequence such that $d_4 \geq 2$, $\sum_i d_i > 2n$
and $\multipl_{odd} + \multipl_1 \geq 2$.
%
There is a bi-cactus realization of $d$ if and only if $m \leq \floor{\frac{4(n - 1) - \beta}{3}}$.
\end{theorem}
\begin{proof}
If there is a bi-cactus realization of $d$, then $m \leq \floor{\frac{4(n-1) - \beta}{3}}$ by 
\Cref{thm:bicactusbeta}.

For the other direction, suppose that $m \leq \floor{\frac{4(n - 1) - \beta}{3}}$.
If $\multipl_{odd} < \multipl_1$, 
then \Cref{thm:corebicactus} implies that there is a realization of $d$ as a core bi-cactus.

Suppose that $\multipl_{odd} \geq \multipl_1$. That is, $\beta = (\multipl_1 + \multipl_{odd})/2$.
We prove the claim by induction on $\multipl_1 + \multipl_{odd}$ and on $\multipl_{odd}$.
In the base case, there are two options.
\begin{compactitem}
\item $\sum_i d_i = 2n$ and $d_4 \geq 2$. Then there exists a bi-unicyclic realization of $d$
          due to \Cref{thm:bipseudotree}.
\item $d_i$ is even for every $n$, $m$ is even, and $n \geq 4$. 
         Since $\sum_i d_i \leq 2\floor{\frac{2(n-1)}{3}}$,
        it follows by \Cref{thm:bipartite-bridgeless} that there exists a realization for $d$
        as a bridge-less bi-cactus.
\end{compactitem}

For the inductive step, there are two cases.
%
%
First, supposed that $\multipl_1 > 0$.
Since $\multipl_{odd} \geq \multipl_1$, the sequence $d$ must contain an odd number $d_j \geq 3$.
Let $d'$ be the sequence which is obtained by subtracting 1 from $d_j$ and removing $n$.
Observe that $\multipl'_{odd} = \multipl_{odd} - 1 \geq \multipl_1 - 1 = \multipl'_1$,
$d'_4 \geq 2$,
\[
\sum_i d'_i = \sum_i d_i - 2 > 2n - 2 = 2n'
~,
\]
and that
\begin{align*}
\sum_i d'_i
=    \sum_i d_i - 2
& \leq 2\floor{ \frac{4(n - 1) - \half(\multipl_1 + \multipl_{odd})}{3} } - 2 \\
& =    2\floor{ \frac{4(n - 1) - \half(\multipl_1 + \multipl_{odd}) - 3}{3} } 
 =    2\floor{ \frac{4(n' -1) - \half(\multipl'_1 + \multipl'_{odd})}{3} }
~.
\end{align*}

First, suppose that $\multipl'_{odd} \geq 1$ or $m'$ is even.
By the induction hypothesis there is a bi-cactus realization $G'$ of $d'$.
We get a bi-cactus realization $G$ of $d$ by adding the edge $(j,n)$.

Next, assume that $\multipl'_{odd} = \multipl'_1 = 0$ and $m'$ is odd.
In this case, we construct a sequence $d^*$ as in the first case of \Cref{thm:corebicactus}.
By the inductive hypothesis there is a realization $G^*$ of $d^*$ as a bi-cactus.
We obtain a bi-cactus realization $G$ of $d$ by adding the edges $(j,n-1)$ and $(n-1,n)$.


The second case is when $\multipl_1 = 0$. In this case, $\multipl_{odd} \geq 2$.
Let $d_j \geq 3$ be the smallest odd number in $d$.
Since $\sum_i d_i \leq 2\floor{\frac{4(n - 1) - \beta}{3}}$, it must be that $d_{n-2} = 2$.
Let $d'$ be the sequence which is obtained by removing $n-2$, $n-1$, and $n$ and
subtracting 2 from $d_j$.
Observe that $\multipl'_{odd} + \multipl'_1 = \multipl_{odd} + \multipl_1$.
In particular, if $d_j=3$, then $\multipl'_1=1$, and otherwise $\multipl'_1=0$.
Also,
\[
\sum_i d'_i = \sum_i d_i - 8 > 2n - 8 = 2n' - 2
~,
\]
namely, $\sum_i d'_i \geq 2n'$.
In addition,
\begin{align*}
\sum_i d'_i
=    \sum_i d_i - 8
& \leq 2\floor{ \frac{4(n - 1) - \half(\multipl_1 + \multipl_{odd})}{3} } - 8 \\
& = 2\floor{ \frac{4(n - 1) - \half(\multipl_1 + \multipl_{odd}) - 12}{3} } 
 = 2\floor{ \frac{4(n' -1) - \half(\multipl'_1 + \multipl'_{odd})}{3} }
~.
\end{align*}
%
Observe that since $m > n$, it must be that $n \geq 8$. Hence, $n' \geq 5$ and 
$d'_4 \geq 2$.
By the induction hypothesis $d'$ has a realization as a bi-cactus $G'$.
We obtain a realization $G$ for $d$ by adding a cycle of the vertices $j$, $n-2$, $n-1$, and $n$.
\end{proof}

The algorithm which is described in the proof of \Cref{thm:bicactus} works
similarly to the one for cacti.
Also, recall that the case where $\beta = 0$ (i.e., $\multipl_1 = \multipl_{odd} = 0$) is covered by \Cref{cor:bi-bridgeless}.

\begin{corollary}
Let $d$ be a degree sequence such that $n \geq 3$, 
$\beta > 0$ and $n - 1 \leq m \leq \floor{\frac{4(n - 1) - \beta}{3}}$.
There is a linear time algorithm that computes a bi-cactus realization of $d$,
where all cycles except maybe one are of length $4$.
\end{corollary}


\section{Forcibly Bi-Cactus Graph Realization}
\label{sec:forcibly}

As mentioned in the introduction a characterization for forcibly cactus graphs was given in~\cite{Rao81cactus}.
Furthermore,  a characterization for forcibly bipartite graphs was given in~\cite{BBPR23}.
Hence, forcibly bi-cactus sequences can be identified by obtaining the intersection of the above two  
characterizations.
However, there is a simpler approach. We observe that the characterization of sequences which are
forcibly \emph{connected} bipartite applies to forcibly bi-cactus sequences.

The following result that was presented in~\cite{BBPR23}:

\begin{theorem}[\cite{BBPR23}]
\label{thm:FCB}
A graphic sequence $d$ is forcibly connected bipartite if and only if 
\begin{inparaenum}[(i)]
\item $d = (2^4)$, or
\item $d = (k,h,1^{n-2})$, for $2 \leq h \leq k$ and $h+k = n$.
\end{inparaenum}
\end{theorem}

It is not hard to verify that the sequences that appear in \Cref{thm:FCB} are 
either a cycle of four vertices or two stars whose centers are connected
as shown in \Cref{fig:FCB}

\begin{figure}[t]
\centering
\begin{subfigure}[t]{0.35\textwidth}
\centering
\begin{tikzpicture}[scale=0.5]
\begin{scope}[every node/.style={fill=blue,shape=circle},
every edge/.style={draw=black,ultra thick}]

\node(A) at (0,0) {};
\node(B) at (0,2) {};
\node(C) at (2,2) {};
\node(D) at (2,0) {};

\path (A) [-] edge (B);
\path (B) [-] edge (C);
\path (C) [-] edge (D);
\path (D) [-] edge (A);
\end{scope}
\end{tikzpicture}
\caption{Unique realization of $(2^4)$.}
\end{subfigure}
\hspace{20pt}
\begin{subfigure}[t]{0.35\textwidth}
\centering
\begin{tikzpicture}[scale=0.5]
\begin{scope}[every node/.style={fill=blue,shape=circle},
every edge/.style={draw=black,ultra thick}]

\node(A) at (2,1) {};
\node(B) at (0,-0.5) {};
\node(C) at (0,0.5) {};
\node(D) at (0,1.5) {};
\node(E) at (0,2.5) {};
\node(F) at (4,1) {};
\node(G) at (6,0) {};
\node(I) at (6,1) {};
\node(J) at (6,2) {};

\path (A) [-] edge (B);
\path (A) [-] edge (C);
\path (A) [-] edge (D);
\path (A) [-] edge (E);
\path (A) [-] edge (F);
\path (F) [-] edge (G);
\path (F) [-] edge (I);
\path (F) [-] edge (J);
\end{scope}
\end{tikzpicture}
\caption{Unique realization of $(5,4,1^7)$.}
\end{subfigure}
\caption{Depiction of forcibly connected bipartite realizations.}
\label{fig:FCB}
\end{figure}
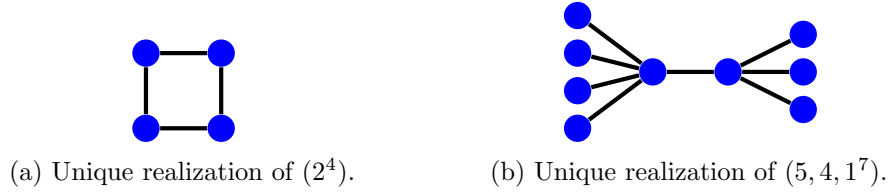

\begin{corollary}
A graphic sequence $d$ is forcibly bi-cactus if and only if 
\begin{inparaenum}[(i)]
\item $d = (2^4)$, or
\item $d = (k,h,1^{n-2})$, for $2 \leq h \leq k$ and $h+k = n$.
\end{inparaenum}
\end{corollary}

Recall that a characterization of forcibly unicyclic sequences was given in~{\cite{DuanTian25}.
 \Cref{thm:FCB} also implies the following result.

\begin{corollary}
A graphic sequence $d$ is forcibly bipartite unicyclic if and only if $d = (2^4)$. 
\end{corollary}



\clearpage

\bibliography{realizations}
\bibliographystyle{abbrv}


\end{document}